\documentclass[conference]{IEEEtran}

\usepackage[T1]{fontenc}
\usepackage{amsmath,amssymb,amsfonts}
\usepackage{newtxtext,newtxmath}
\newtheorem{definition}{Definition}
\newtheorem{theorem}{Theorem}
\usepackage{algorithmic}
\usepackage{algorithm}
\usepackage{graphicx}
\usepackage{textcomp}
\usepackage{xcolor}
\usepackage{booktabs}
\usepackage{multirow}
\usepackage[hidelinks]{hyperref}
\usepackage{cleveref}
\usepackage{microtype}
\usepackage{tikz}
\usetikzlibrary{arrows.meta,positioning,shapes.geometric,fit}
\usepackage{pgfplots}
\pgfplotsset{compat=1.18}
\graphicspath{{./}}

\begin{document}

\title{Confidence-Gated Admission for Hardware Prefetching:\\When the Gate Matters More Than the Predictor}

\author{
\IEEEauthorblockN{Youssef Majdane}
\IEEEauthorblockA{Algorithmica Solutions\\youssef@algorithmicasolutions.com}
\and
\IEEEauthorblockN{Simone Jarno Casartelli}
\IEEEauthorblockA{Algorithmica Solutions\\simone@algorithmicasolutions.com}
\and
\IEEEauthorblockN{Enrico Lopedoto}
\IEEEauthorblockA{Algorithmica Solutions\\info@algorithmicasolutions.com}
}

\maketitle

\begin{abstract}
Learned cache prefetchers are typically evaluated against classical
predictors that always issue requests, confounding the prediction model
with the admission policy. We disentangle these variables with matched
controls: the same admission gate is applied to both a 257-parameter
online MLP and a classical stride predictor. The neural advantage
vanishes; the MLP is indistinguishable from gated stride on random
traffic and slower on most regular streams. The gate itself is
architecturally useful independent of the predictor: on twenty SPEC
CPU2017 programs in native ChampSim, it removes 35\% of prefetches and
improves accuracy from 11\% to 15\%, but DRAM reads change by only
$-0.07\%$---demonstrating that proxy metrics do not predict endpoint
behavior. We prove gate-closed execution reproduces the no-prefetch
baseline exactly. The gate matters more than the predictor,
and better proxies do not imply better endpoints.
\end{abstract}


\section{Introduction}
\label{sec:intro}

The memory wall has not moved: an L1 miss costs tens of cycles, and for
memory-bound code that cost dominates. Prefetching attacks the problem by
predicting the next addresses and bringing their lines into the cache before
they are requested~\cite{wulf1995,chen1995}. Decades of work produced a
robust toolbox for the easy cases: next-line prefetching, stride
detection~\cite{chen1995}, signature-path lookahead~\cite{kim2016spp},
spatial/best-offset units~\cite{michaud2016bop}, and Bingo spatial
prefetching~\cite{bingo}. On sequential, strided, and
gather-within-a-stream traffic these mechanisms are close to optimal.

\textbf{The learned-prefetching wave.}
A recent generation of neural prefetchers has pushed beyond heuristics.
Pythia~\cite{pythia} cast prefetching as online reinforcement learning.
TransFetch~\cite{transfetch} applied transformers with fine-grained address
segmentation. Hashemi et al.~\cite{hashemi2018} demonstrated that LSTMs
can learn irregular memory access patterns.
DART~\cite{dart} achieved hardware-practical inference through
tabularization at IPDPS~2024, and Pathfinder~\cite{pathfinder} demonstrated
real-time online learning with delta patterns at ASPLOS~2024.

\textbf{The confounded evaluation.}
These systems share a consequential evaluation pattern: a learned prefetcher
with an implicit or explicit confidence mechanism is compared against a
classical unit that always issues requests. This confounds two independent
variables---the address-prediction model and the admission policy. A gated
learned predictor that stays silent on random traffic looks ``safer'' than an
always-on stride unit, but the safety comes from the gate, not from the
neural network. Prior throttling mechanisms such as Feedback-Directed
Prefetching~\cite{fdp2007} adjust aggressiveness based on accuracy feedback,
but apply a single policy to one predictor at a time.
No prior work we have found uses a matched-control design that applies the
\emph{same} admission policy to both a learned and a classical predictor
before drawing conclusions about the predictor's value.

\textbf{Why it matters now.}
Real workloads are rarely all-or-nothing. An LLM inference loop cycles
between sequential KV-cache appends and far attention
gathers~\cite{kwon2023vllm,pope2023}; a retrieval agent alternates long
context scans with scattered lookups; embeddings are dense rows reached by
sparse indices. A prefetcher that cannot distinguish ``regular, with the
occasional far access'' from ``pure noise'' will misbehave on exactly these
workloads---and an evaluation that credits the prediction model for what the
gate delivers will misguide architectural decisions.

\textbf{Contributions.} We make four contributions:
\begin{enumerate}
\item \textbf{Methodology.} A matched-control falsification that applies the
same confidence and regularity gates to both the MLP and a stride predictor,
isolating prediction from admission (\cref{sec:matched}).
\item \textbf{Architecture.} A confidence gate (median-error) and a
regularity gate (dominant-delta support) that suppress prefetching on
unlearnable streams, useful independent of the predictor. We prove that
gate-closed execution equals the baseline (Theorem~\ref{thm:safety}) and
verify empirically that the gate closes on all tested unlearnable
streams (\cref{sec:safety}).
\item \textbf{External validation.} A full twenty-program SPEC CPU2017 study
in native ChampSim (eleven development, nine frozen holdout), with
predictions registered before observation, showing that proxy-metric
improvements (fewer prefetches, higher accuracy) do not translate to
endpoint gains (DRAM reads, IPC) (\cref{sec:champsim}).
\item \textbf{Reproducibility.} Every claim is machine-verifiable:
\texttt{verify\_paper.py} checks local results,
\texttt{verify\_research.py} checks the full experimental matrix, and
\texttt{champsim/verify\_results.py} reconstructs native aggregates from
raw logs (Reproducibility Statement, below).
\end{enumerate}

\section{Background and Related Work}
\label{sec:related}

\subsection{Classical Hardware Prefetching}

Hardware prefetching is dominated by pattern-matching heuristics. Stride and
sequence detection~\cite{chen1995} handle regular access patterns. SPP
(Signature Path Prefetcher)~\cite{kim2016spp} uses compressed address
signatures to look ahead along correlated paths. Best-offset
prefetching~\cite{michaud2016bop} learns the optimal constant offset that
maximizes timeliness. Bingo~\cite{bingo} captures spatial correlations
within physical pages. These excel on their target streams and are known to
be fragile when their assumptions break---the gap we formalize with the
no-prefetch baseline comparison.

\subsection{Learned Prefetching}

Hashemi et al.~\cite{hashemi2018} showed that LSTMs can learn irregular
memory access patterns, establishing neural prefetching as a research
direction. Pythia~\cite{pythia} cast prefetching as online reinforcement
learning, training a customizable agent with cache-hit reward signals at
MICRO~2021; its 64\,KB state table and offline training on
ChampSim~\cite{champsim} traces achieve strong coverage.
TransFetch~\cite{transfetch} applied attention-based address segmentation
with variable-degree prefetching at Computing Frontiers~2022.
DART~\cite{dart} achieved hardware-practical inference by converting
attention models into compact table hierarchies via distillation
(IPDPS~2024). Pathfinder~\cite{pathfinder} demonstrated practical real-time
online learning with delta patterns at ASPLOS~2024.
Berti~\cite{berti} is an accurate local-delta prefetcher that achieves
high coverage without global history, demonstrating that well-engineered
classical predictors remain a strong baseline---a point our
matched-control results reinforce.

\textbf{Positioning.} We share the online-learning ingredient with Pythia,
the delta-based approach with Pathfinder and Berti, and the deployment
pragmatism with DART. Our distinction is methodological: none of these works
applies a matched admission policy to both the learned and classical
predictor before attributing performance differences to the model. We
introduce and verify a safety property (Theorem~\ref{thm:safety}), which
none of them declares, and then show via matched controls that the property
belongs to the gate rather than to the network.

\subsection{LLM Inference Memory Optimization}

Our motivating application sits at the center of memory-bound serving: LLM
inference. Attention and the KV cache dominate its memory
traffic~\cite{pope2023,kwon2023vllm}; the community already treats the KV
cache as a first-class storage object with paged
layouts~\cite{kwon2023vllm,kvcache_streamed}, quantization~\cite{kvcache_quant},
batched serving~\cite{orca,sarathi}, and structured-generation
frameworks~\cite{sglang}. Speculative
decoding~\cite{leviathan2023} showed that prediction with verification
works, and recent systems prefetch KV blocks or MoE experts with learned
confidence~\cite{kvcache_lookahead,apex,kvcache_moe}---conceptually our
gate at a coarser granularity, aimed at throughput rather than a no-harm
guarantee.

\subsection{Safety and Confidence in Prefetching}

Safe Policy Improvement with Baseline Bootstrapping (SPIBB)~\cite{laroche2019}
provides formal guarantees that a learned policy is never worse than the
baseline. Learned index structures with B-tree fallback~\cite{kraska2018}
implement a structurally similar pattern: use the learned model when confident,
fall back otherwise. Our median-error gate is the prefetching instantiation of
this predict-then-verify architecture, with the additional empirical finding
that the gate's value is independent of the predictor behind it.
PPF~\cite{ppf} corroborates this independently: a perceptron confidence
gate on SPP yields 3.78\% IPC improvement (ISCA~2019), attributable to
the filter rather than the underlying predictor.

\section{System Design}
\label{sec:design}

\subsection{Simulator Model}
\label{sec:simulator}

We use an in-order, latency-accounting scalar core. This is deliberate:
every number is a deterministic function of the modeled latencies, and a
full pipeline front-end would add complexity without changing qualitative
conclusions. One cycle fetches an instruction. The default L1 is fully
associative, 32~lines of 8~words each, LRU; the matched-control
falsification (\cref{tab:matched}) uses a timed 8-way configuration to
stress realistic contention. A load/store that hits costs one cycle; a
miss stalls the core for $L = 40$ cycles. A prefetch inserts a line at zero
cycle cost, modeling otherwise-idle DRAM bandwidth. Stores are write-through
with a 16-entry store buffer that stalls only on overflow. Arithmetic
latencies are 1/3/8~cycles for ADD/MUL/DIV. Every run records cycles, hits,
misses, hit rate, prefetches issued vs.\ used, IPC, and per-category cycle
totals, persisted with full configuration and provenance.

\subsection{Delta Encoding and MLP Predictor}
\label{sec:predictor}

Let $a_t$ be the address of memory access $t$. We predict the next delta,
\begin{equation}
d_t = a_t - a_{t-1}, \qquad \hat{d}_{t+1} = f_\theta(\mathbf{x}_t),
\end{equation}
from a 6-dimensional feature vector:
\begin{equation}
\mathbf{x}_t = \bigl[\,\text{opcode},\;\text{pc},\;\text{addr},\;\text{delta},\;\text{delta}_{t-1},\;\text{hit/miss}\,\bigr].
\end{equation}
The predictor $f_\theta$ is a 6--32--1 ReLU MLP with Adam: 257~parameters,
${\approx}26\,\mu$s per prediction in pure NumPy. Every memory access closes
the previous training sample and issues one prediction over a 64-sample
replay buffer fitted every 16~samples. Training targets and predictions are
clamped to ${\pm}16$~words: on gather-heavy streams an attention jump moves
thousands of words, and letting that outlier into the regression averages
away the sequential delta the model needs. \Cref{fig:architecture} shows the
full decision and training loop; \cref{alg:gated} gives the per-access
pseudocode.

\begin{figure}[t]
\centering
\begin{tikzpicture}[
  node distance=0.6cm and 0.8cm,
  box/.style={draw, rounded corners, minimum height=0.7cm, minimum width=1.6cm,
              font=\scriptsize, align=center, fill=blue!8},
  decision/.style={draw, diamond, aspect=2, minimum width=1.2cm,
                   font=\scriptsize, align=center, fill=yellow!15},
  io/.style={draw, trapezium, trapezium left angle=70, trapezium right angle=110,
             minimum width=1.2cm, font=\scriptsize, align=center, fill=green!10},
  arrow/.style={-{Stealth[length=2mm]}, thick},
]
\node[io] (input) {$a_t$\\(6 dims)};
\node[box, right=of input] (mlp) {MLP $f_\theta$\\$\hat{d}$};
\node[decision, right=of mlp] (gate) {confident?};
\node[box, above right=0.3cm and 0.7cm of gate] (pf) {prefetch\\$a_t + \hat{d}$};
\node[box, below right=0.3cm and 0.7cm of gate] (noop) {no-op};

\node[box, below=1.2cm of input] (replay) {replay ($\leq$64)\\clamp $d_t \to \pm 16$};
\node[box, right=of replay] (refit) {refit $f_\theta$\\every 16};
\node[box, right=of refit] (errhist) {error history\\median $|\hat{d}-d|$\\$W\!=\!48$};

\draw[arrow] (input) -- (mlp);
\draw[arrow] (mlp) -- (gate);
\draw[arrow] (gate) -- node[above,font=\tiny]{yes} (pf);
\draw[arrow] (gate) -- node[below,font=\tiny]{no} (noop);
\draw[arrow, dashed] (input) |- node[left,font=\tiny,pos=0.25]{close sample} (replay);
\draw[arrow, dashed] (replay) -- (refit);
\draw[arrow, dashed] (refit) -- (errhist);
\draw[arrow, dashed] (errhist) -- (gate);
\end{tikzpicture}
\caption{Architecture of the gated delta predictor. Top: the decision
path---every access is encoded, the MLP proposes a delta, and the gate
either accepts (prefetch) or falls back to no-op. Bottom (dashed): the
online learning loop and error history that feeds the gate.}
\label{fig:architecture}
\end{figure}
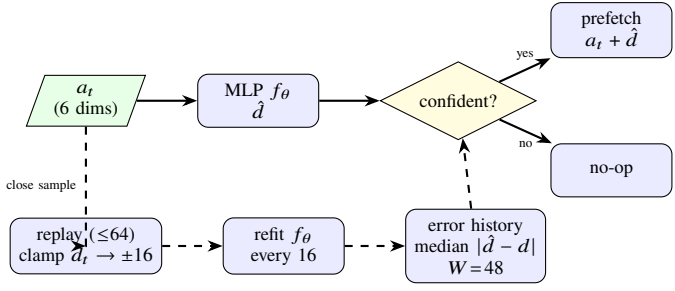

\begin{algorithm}[t]
\caption{Gated delta predictor (per memory access $a_t$)}
\label{alg:gated}
\begin{algorithmic}[1]
\REQUIRE predictor $f_\theta$; replay $R$ ($\leq 64$); gate window $W = 48$; sample counter $t$
\STATE $d_t \gets a_t - a_{t-1}$
\STATE $R.\text{append}(\mathbf{x}_{t-1}, \text{clamp}(d_t))$
\STATE $t \gets t + 1$
\IF{$t \bmod 16 = 0$}
  \STATE fit $f_\theta$ on $R$
\ENDIF
\STATE $\hat{d}_{t+1} \gets f_\theta(\mathbf{x}_t)$
\STATE $e_t \gets |\hat{d}_t - d_t|$
\STATE push $e_t$ into error window $E$ (size $W$)
\IF{$\text{gate}(E, \hat{d}_{t+1})$ = confident}
  \STATE prefetch at $a_t + \hat{d}_{t+1}$
\ELSE
  \STATE no-op
\ENDIF
\end{algorithmic}
\end{algorithm}

\subsection{Confidence Gate}
\label{sec:confidence_gate}

The confidence gate decides whether the current prediction is trustworthy
enough to act on. Let $E = (e_1, \ldots, e_W)$ be the window of recent
absolute prediction errors. The gate is:
\begin{equation}
\text{gate}(E, \hat{d}) = \begin{cases}
\text{confident} & \text{if } \text{median}(E) < \tau \cdot |\hat{d}| \\
\text{suppress}  & \text{otherwise}
\end{cases}
\label{eq:gate}
\end{equation}
where $\tau = 1.0$, selected from $\{0.5, 1.0, 2.0\}$ on development seeds
0--9: $\tau = 0.5$ closes the gate too aggressively on learnable streams;
$\tau = 2.0$ opens it too liberally and degrades safety on random traffic.
The median is robust to outliers from attention gathers; scaling by
$|\hat{d}|$ makes the threshold adaptive. On unlearnable streams the median
error stays high, the gate remains closed, and the predictor issues no
prefetches. On learnable streams the error drops quickly and the gate opens.
When traffic switches phase, the error window adapts within $W$ accesses.

\textbf{Hardware cost.} The confidence gate requires a running median over a
$W\!=\!48$ entry error window---an awkward structure for hardware, typically
requiring a sorted insertion buffer or approximate median network. Combined
with the MLP's ${\approx}26\,\mu$s NumPy inference latency, the confidence
gate is practical only as a simulation-level instrument. The regularity gate
(\cref{sec:regularity_gate}) is the hardware-viable alternative; the native
ChampSim evaluation validates it exclusively.

\subsection{Regularity Gate}
\label{sec:regularity_gate}

Prediction-error gating can reject repeated multimodal streams where a
dominant pattern exists but prediction error remains high. We therefore
define a predictor-independent regularity measure:
\begin{equation}
R_t = \frac{\max_d |\{i : d_i = d\}|}{W_r}.
\label{eq:regularity}
\end{equation}
After eight deltas, the gate admits if and only if $R_t \geq 0.35$ for
$W_r = 16$. Window and threshold were chosen on development seeds 0--9; all
confirmatory results use disjoint seeds 10--39.

An exact incremental implementation requires a 16-entry signed 64-bit FIFO
and a 16-entry associative histogram: 2,153 state bits (269.1 bytes, 0.41\%
of the DPC-3 64-KiB budget). This is storage accounting; comparator and
update logic are not costed.

\subsection{Matched Classical Controls}
\label{sec:matched}

The matched-control design is the methodological core of this work. We wrap
the same admission gates around a classical stride predictor
(which predicts the previous delta), creating a \emph{gated stride} control.
Both the MLP and gated stride see the same stream, use the same gate
parameters, and train/adapt while the gate is closed. The only variable that
differs is the address-prediction model.

This eliminates the admission confound: any performance difference between
gated MLP and gated stride is attributable to the predictor, not to the
decision of whether to prefetch.

\section{Safety Analysis}
\label{sec:safety}

\subsection{\texorpdfstring{Gate-Closed Safety (Theorem~\ref{thm:safety})}{Gate-Closed Safety (Theorem 1)}}

The safety argument has two layers: a theorem about the system design and an
empirical verification that the theorem's precondition holds on tested
workloads.

\begin{definition}[Safety]
\label{def:safety}
A gated prefetcher satisfies the safety property on an execution if
$\text{cyc}_{\text{gated}}(w, s) \leq \text{cyc}_{\text{baseline}}(w, s).$
\end{definition}

\begin{theorem}[Gate-Closed Safety]
\label{thm:safety}
If the confidence gate suppresses every prefetch during execution on
workload~$w$ with seed~$s$, then
$\text{cyc}_{\text{gated}}(w, s) = \text{cyc}_{\text{baseline}}(w, s).$
\end{theorem}

\noindent\textit{Proof.}
By induction on memory accesses. At $t{=}0$ the cache states are identical.
At access~$t$, because the gate is closed, no prefetch line is inserted; the
gated prefetcher executes the same load/store on the same cache state as the
baseline, producing the same hit/miss outcome and latency. Since no cache
mutation occurs beyond what the baseline performs, the two executions remain
in lock-step for all~$t$. \hfill$\square$

\subsection{Scope and Limits}
\label{sec:safety_limits}

The theorem is deliberately simple: it says only that a prefetcher issuing no
prefetches behaves identically to no prefetcher. The non-trivial claim is
that a 257-parameter model's error statistics \emph{reliably trigger} gate
closure on unlearnable streams. The theorem makes the safety argument
structural---independent of workload, seed, or model quality---so the
empirical question reduces to: does the gate close?

\Cref{def:safety} defines safety as
$\text{cyc}_\text{gated} \leq \text{cyc}_\text{baseline}$, permitting a
gate-open execution to qualify if it accelerates.
Theorem~\ref{thm:safety} proves the stronger gate-closed equality ($=$),
which implies~($\leq$).

\textbf{What the theorem does not say.} Gate-closed safety guarantees no
worse than \emph{no prefetching}. It does not guarantee no worse than
\emph{optimal classical prefetching}. On streams where stride or best-offset
achieves $9$--$11\times$ speedup, the gated NN's safety means it does
nothing---which is safe but unhelpful. The theorem is about the gate
mechanism, not about the predictor's quality.

\section{Experimental Setup}
\label{sec:setup}

\subsection{Workloads}
\label{sec:workloads}

All memory-workload footprints far exceed the L1, so misses genuinely
dominate and prefetching has room to win.

\textbf{Synthetic (deterministic):}
\texttt{sequential\_read} (stride~1),
\texttt{strided\_read} (step~8),
\texttt{mixed\_read\_write} (alternating stores/loads),
\texttt{random\_read} (linear-congruential generator),
\texttt{phased\_switch} (alternating predictable and noisy phases), and
\texttt{arithmetic\_mix} (no memory traffic).

\textbf{ML/agent patterns:}
\texttt{kv\_cache\_append} (sequential KV writes plus periodic attention
gathers), \texttt{embedding\_lookup} (dense rows via sparse indices),
\texttt{agent\_rag} (long context scans plus scattered retrieval bursts),
and \texttt{token\_stream}.

\textbf{Real algorithm traces} (merge sort, row-major matrix multiply,
binary search) are captured element-by-element from executing code and
replayed verbatim.

\subsection{Statistical Methodology}

Five local configurations: no prefetching (baseline), stride detector,
next-line, best-offset, and the gated NN. Each synthetic workload runs
30~independent seeds; ML/agent and real traces run 8~seeds. We report
speedup $= \text{cyc}_\text{baseline} / \text{cyc}_\text{config}$, L1 hit
rate, and paired Wilcoxon signed-rank tests for statistical significance.
For matched-control comparisons we additionally report the paired cycle
ratio $B(A, C) = \text{mean}_s\, C_s / A_s$, where $B > 1$ favors
candidate~$A$ over control~$C$, with exact sign tests under Holm
correction.

\subsection{Machine Configuration Sweep}

To verify that conclusions are not artifacts of a single parameter set, we
sweep L1 capacity (16--64~lines), line size (4--16~words), and DRAM latency
(20--80~cycles). The safety property holds across all tested settings:
NN~$=$~baseline on random for every configuration, while stride drops below
baseline on all of them (\cref{fig:sensitivity}).

\subsection{ChampSim Evaluation Protocol}
\label{sec:champsim_protocol}

The \emph{regularity gate} (\cref{sec:regularity_gate}) is implemented as a native
ChampSim runtime L1D prefetcher module
using official \texttt{develop} commit \texttt{e6530c7}~\cite{champsim}.
The confidence gate remains a local-simulator instrument only; the native
study evaluates the regularity gate on the classical stride predictor.
One binary selects no~L1D prefetch, official \texttt{ip\_stride},
gate-disabled matched stride, or gated stride through JSON configuration.
The gate-disabled module must match official \texttt{ip\_stride} exactly on
all recorded output fields.

We evaluate on twenty SPEC CPU2017 programs---the full set of distinct
programs with publicly available DPC-3 SimPoints~\cite{dpc3}---each using
its highest-weight SimPoint. Eleven programs served as the development set;
the remaining nine were frozen as a holdout before observation, with
predictions registered in \texttt{EXTERNAL\_VALIDATION.md} before any
holdout simulation. The primary run warms 50~million and measures
200~million retired instructions with 15-way parallelism on a 16-core host.
One program is one independent unit; SimPoint weights select a trace and are
not replicate counts.

\section{Results: Local Simulator}
\label{sec:local_results}

\subsection{Gate-Closed Safety on Unlearnable Streams}

\begin{table}[t]
\centering
\caption{Canonical run, mean over 30 seeds (cycles $\pm$ 95\% CI, speedup
$\pm$ std). Speedup is the mean of per-seed ratios (baseline cycles / config
cycles), not the ratio of the shown means.
On \texttt{random\_read} the gated NN achieves near-baseline
performance ($1.000\times$) while stride completes in more cycles
($0.98\times$). On \texttt{mixed\_rw} the gate is overly conservative:
the NN stays near baseline while stride achieves $8.69\times$.}
\label{tab:canonical}
\footnotesize
\setlength{\tabcolsep}{3pt}
\begin{tabular}{@{}l rr rr rr@{}}
\toprule
& \multicolumn{2}{c}{baseline} & \multicolumn{2}{c}{stride} & \multicolumn{2}{c}{NN} \\
\cmidrule(lr){2-3} \cmidrule(lr){4-5} \cmidrule(lr){6-7}
Workload & cyc & hit & cyc & sp & cyc & sp \\
\midrule
seq\_read     & 13\,750 & .875 & 4\,039 & 3.40$\times$ & 4\,161$\pm$16 & 3.30$\pm$.04$\times$ \\
strided\_read & 82\,000 & .000 & 4\,117 & 19.92$\times$ & 9\,240$\pm$528 & 9.10$\pm$1.50$\times$ \\
mixed\_rw     & 43\,984 & .500 & 5\,062 & 8.69$\times$ & 43\,867$\pm$0 & 1.003$\pm$.000$\times$ \\
random\_read  & 73\,952 & .103 & 75\,625 & 0.98$\times$ & 73\,931$\pm$35 & \textbf{1.000}$\pm$.001$\times$ \\
phased\_sw    & 43\,629 & .492 & 39\,906 & 1.09$\times$ & 39\,565$\pm$72 & 1.10$\pm$.01$\times$ \\
\bottomrule
\end{tabular}
\end{table}

\begin{figure}[t]
\centering
\includegraphics[width=\columnwidth]{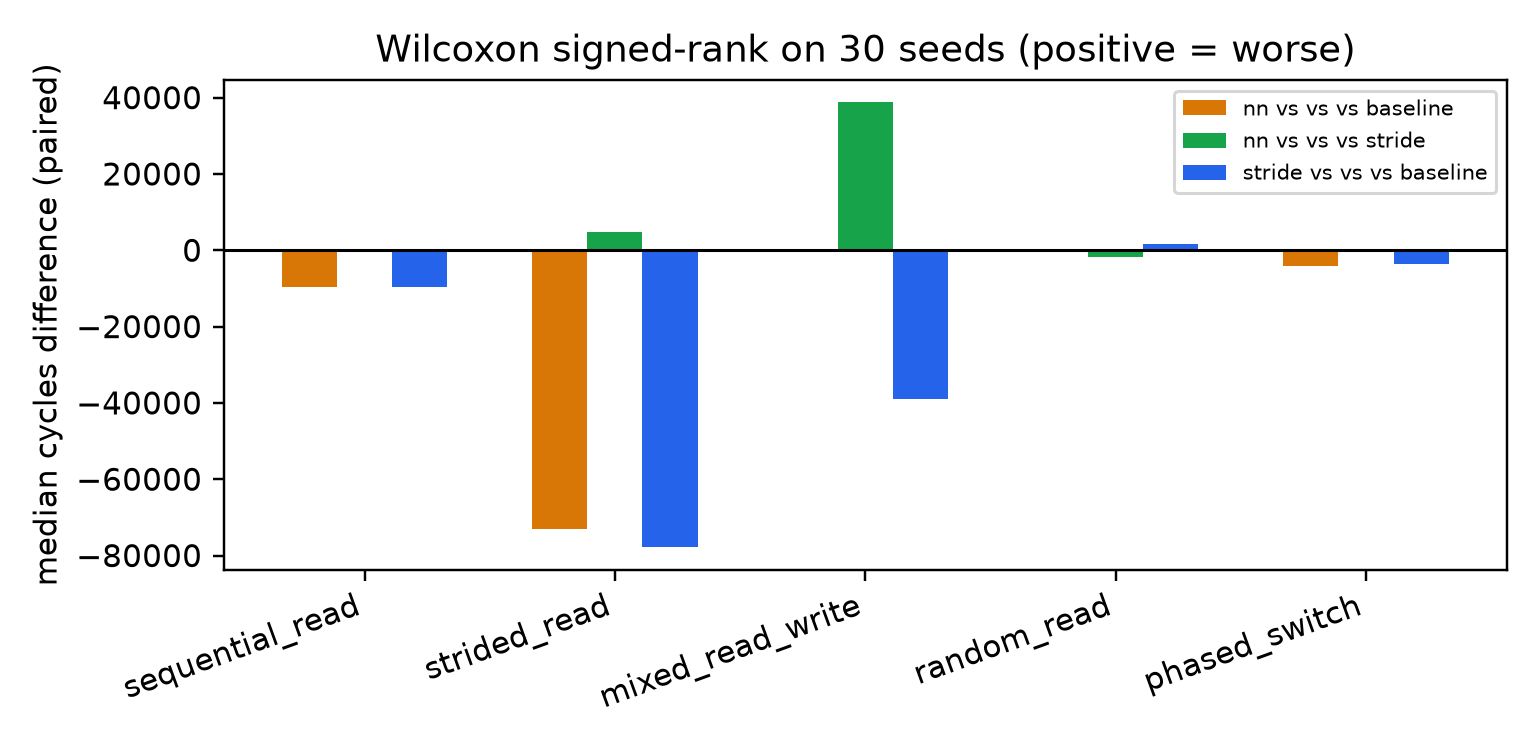}
\caption{Paired Wilcoxon results (30 seeds): median difference in total
cycles per (workload, comparison). A negative value means the first
configuration completes in fewer cycles.}
\label{fig:wilcoxon}
\end{figure}

\begin{figure}[t]
\centering
\includegraphics[width=\columnwidth]{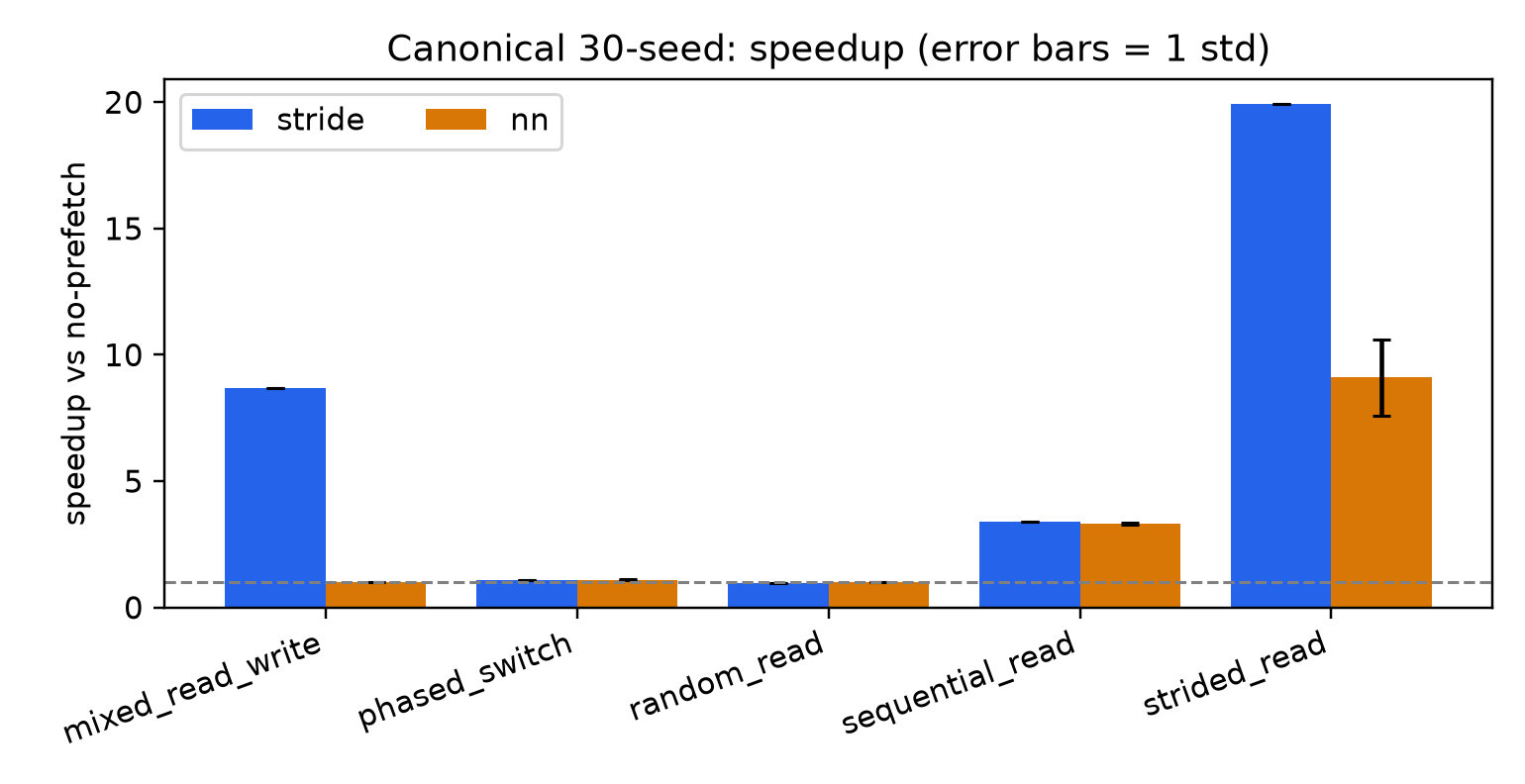}
\caption{Speedup vs.\ no-prefetch (30 seeds, error bars $= 1$ std). The NN
tracks stride on regular streams and rests exactly at $1.0\times$ on random.}
\label{fig:speedup}
\end{figure}

On \texttt{random\_read}, over 30 independent seeds the gated NN achieves
near-baseline performance (mean speedup $1.000\times$, range
$[0.999,\,1.002]$, hit~0.103 vs.\ baseline 0.103); the gate suppresses the
vast majority of prefetches, with residual accuracy of 9\% over negligible
coverage (${<}0.04\%$)
(\cref{tab:canonical,fig:speedup,fig:wilcoxon}). The residual is
statistically detectable (Wilcoxon $p = 0.008$) but practically negligible
(${<}0.03\%$ of total cycles). Under the stride configuration the same
workload completes in \emph{more} cycles ($0.98\times$, paired Wilcoxon
$p < 0.001$) with a lower hit rate (0.082).
Theorem~\ref{thm:safety} guarantees exact baseline reproduction when the
gate closes completely; in practice the gate nearly closes, and the
residual effect is indistinguishable from noise at the application level.

On \texttt{mixed\_rw} the gate is similarly conservative: the NN achieves
only $1.003\times$ while stride reaches $8.69\times$, illustrating that the
gate can suppress prefetching on patterns that are easily learnable by
classical heuristics.

We define margin recovery as $(s_\text{NN} - 1) / (s_\text{stride} - 1)$.
On smooth streams: 96\% (sequential) and 43\% (strided). On
\texttt{mixed\_rw} the gate closes almost entirely, recovering less than
1\% of the stride margin.

\subsection{Phase Switching Recovery}

The \texttt{phased\_switch} workload alternates predictable and noisy phases.
Per-phase hit rates (10~seeds): predictable phases 0.883 (baseline), 0.995
(stride), 0.922 (NN); noisy phases 0.104 (baseline), 0.085 (stride), 0.105
(NN). In noisy phases stride's hit rate falls below baseline; the NN's
equals it---the gate was closed. Aggregate: NN $1.10\times$, stride
$1.09\times$---the gate's closure during noisy phases gives the NN a small
but significant edge (Wilcoxon $p < 0.001$).

\subsection{ML/Agent Patterns}

\begin{table}[t]
\centering
\caption{Prefetcher comparisons on ML/agent patterns and real traces, mean
over 8 seeds. Bold marks head-to-head wins.}
\label{tab:mlagent}
\begin{tabular}{@{}l cccc@{}}
\toprule
Workload & stride & nextline & bestoff & NN \\
\midrule
token\_stream    & 3.40$\times$ & 3.40$\times$ & 3.40$\times$ & 3.26$\times$ \\
agent\_rag       & 2.14$\times$ & 2.14$\times$ & 2.13$\times$ & 1.95$\times$ \\
embed\_lookup    & 1.88$\times$ & 2.01$\times$ & 1.90$\times$ & 1.64$\times$ \\
kv\_cache\_app   & 1.97$\times$ & 1.00$\times$ & 2.76$\times$ & \textbf{1.33}$\times$ \\
random\_read     & 0.98$\times$ & 1.00$\times$ & 0.98$\times$ & \textbf{1.000}$\times$ \\
real\_matmul     & 9.63$\times$ & 11.24$\times$ & 8.46$\times$ & 1.00$\times$ \\
real\_mergesort  & 6.86$\times$ & 8.68$\times$ & 9.04$\times$ & 1.00$\times$ \\
\bottomrule
\end{tabular}
\end{table}

\begin{figure}[t]
\centering
\includegraphics[width=\columnwidth]{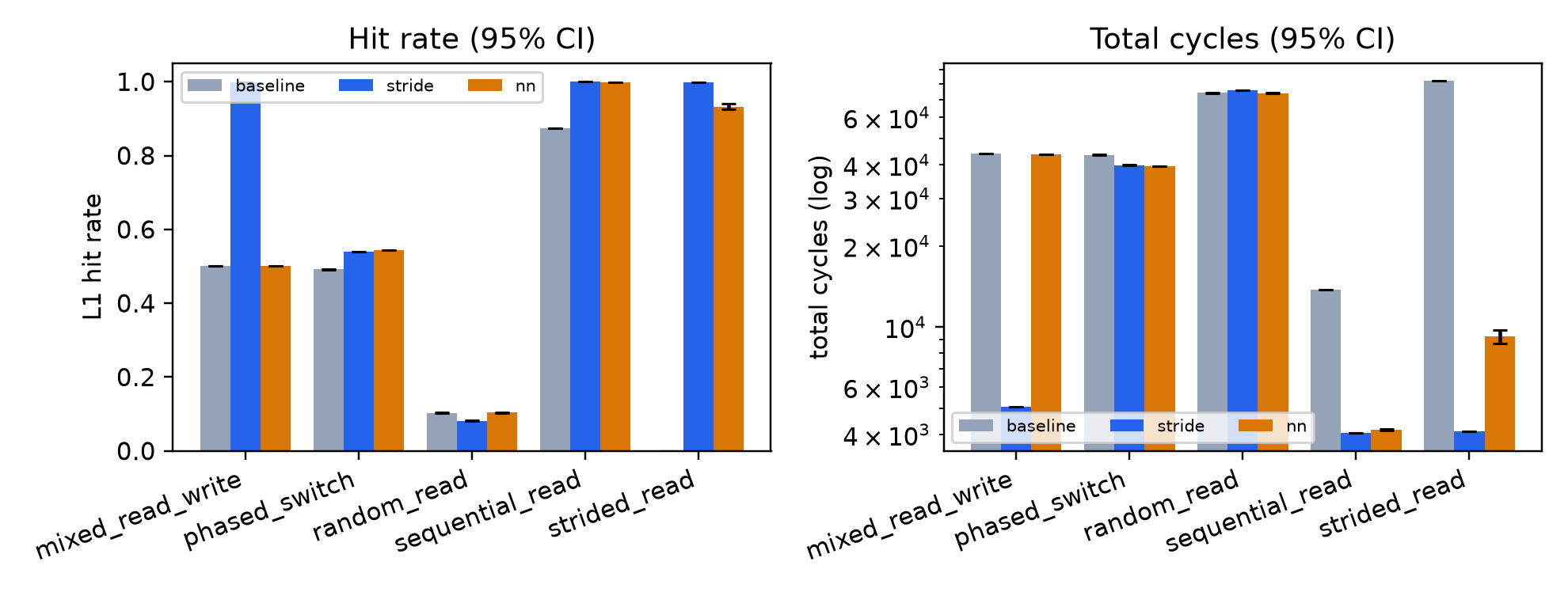}
\caption{Left: L1 hit rate with 95\% CI; \texttt{random\_read} is the only
row where stride drops below baseline. Right: total cycles on a log scale
(95\% CI).}
\label{fig:hitrate}
\end{figure}

\Cref{tab:mlagent} compares all configurations on ML/agent patterns and
real traces (8~seeds). The NN wins head-to-head in two cases: on
\texttt{random\_read} it beats stride and best-offset (both $0.98\times$);
on \texttt{kv\_cache\_append} it beats next-line ($1.33\times$ vs.\
$1.00\times$). On the remaining five workloads the NN underperforms
heuristics: stride or next-line achieve up to $11\times$ higher speedup.
This is the cost of the safety gate's conservatism.

The \texttt{real\_matmul} and \texttt{real\_mergesort} rows expose the
documented limitation: with a 32-word row stride, the replay buffer's
${\pm}16$-word clamp prevents the MLP from representing the pattern, and the
gate correctly identifies unreliable predictions and suppresses prefetching.
Classical mechanisms remain the right tool for dense, large-stride, regular
compute (\cref{fig:hitrate}).

\subsection{Feature Ablation and Gate Design}

\begin{table}[t]
\centering
\caption{Feature ablation, mean of 5 seeds. Removing the absolute address
\emph{increases} speedup; the delta encoding is the critical feature.}
\label{tab:ablation}
\begin{tabular}{@{}l cc@{}}
\toprule
Encoding & strided & mixed \\
\midrule
full          & 8.33$\times$  & 6.15$\times$ \\
no-abs-addr   & 16.96$\times$ & 8.04$\times$ \\
delta-only    & 16.51$\times$ & 8.07$\times$ \\
\bottomrule
\end{tabular}
\end{table}

\Cref{tab:ablation} confirms that the delta encoding is the critical design
choice: removing the absolute address is the only change that
\emph{increases} speedup ($8.3\times \to 17.0\times$ on strided reads).
Replacing the mean error gate with the median and clamping training deltas
to ${\pm}16$~words fixes gather-heavy streams: embedding
$1.00 \to 1.64\times$, KV-cache $1.00 \to 1.33\times$, agent-RAG
$1.84 \to 1.95\times$, with no regression on random or smooth streams.

\subsection{Matched-Control Falsification}
\label{sec:falsification}

\begin{table}[t]
\centering
\caption{MLP vs.\ regularity-filtered stride, local timed 8-way associativity,
prefetch distance $D=16$ cache lines.
$B > 1$ favors the MLP; holdout seeds 10--39.}
\label{tab:matched}
\begin{tabular}{@{}l rr@{}}
\toprule
Workload & $B(\text{MLP}, \text{regstride})$ & Holm result \\
\midrule
random             & 0.999 & n.s. \\
phased             & 0.995 & MLP slower \\
KV-cache append    & 1.101 & MLP faster \\
embedding lookup   & 0.875 & MLP slower \\
agent RAG          & 0.944 & MLP slower \\
strided read       & 0.510 & MLP slower \\
mixed read/write   & 0.128 & MLP slower \\
token stream       & 0.944 & MLP slower \\
\bottomrule
\end{tabular}
\end{table}

\Cref{tab:matched} is the central methodological result. When the same
regularity gate is applied to both the MLP and a stride predictor, the
broad neural advantage disappears. On random traffic the two are
statistically indistinguishable ($B = 0.999$, n.s.). On six of eight
workloads the MLP is significantly slower. The sole MLP win is KV-cache
append ($B = 1.101$), where the bounded regressor preserves a local delta
while ignoring gathers---a conditional mechanism, not evidence of general
superiority.

The conclusion is stark: the safety property observed in
\cref{tab:canonical}---where the ungated NN at $1.000\times$ looked
superior to the ungated stride at $0.98\times$ on random---was a gate
effect, not a prediction effect. With matched admission, the predictor
difference vanishes.

\section{Results: Native ChampSim Validation}
\label{sec:champsim}

\subsection{Correctness Invariant}

The gate-disabled module matches official ChampSim \texttt{ip\_stride}
exactly on every recorded output field (instructions, cycles, cache metrics,
prefetch metrics, DRAM reads/writes) across all twenty programs. The
paired IPC ratio is exactly~1 with twenty ties. The gate comparison is
therefore not an accidental reimplementation artifact.

\subsection{Gate Performance}

\begin{table}[t]
\centering
\caption{Native ChampSim twenty-program study, gate vs.\ matched raw stride.
Reductions are aggregate counts over twenty equal-length programs.}
\label{tab:champsim}
\begin{tabular}{@{}l r@{}}
\toprule
Metric & Gate vs.\ raw stride \\
\midrule
Geometric-mean IPC ratio & 0.996 \\
Bootstrap 95\% CI        & $[0.990,\,0.9997]$ \\
Exact sign test $p$      & 0.481 (n.s.) \\
Win / loss / tie         & 7 / 11 / 2 \\
Accepted L1D prefetches  & $-35.1\%$ (agg.), $-60.5\%$ (med.) \\
L2C prefetch misses      & $-8.4\%$ \\
LLC prefetch misses      & $-7.6\%$ \\
DRAM read requests       & $-0.07\%$ \\
Accepted-prefetch accuracy & $11.0\% \to 15.4\%$ \\
Useful-prefetch retention  & 90.3\% \\
\bottomrule
\end{tabular}
\end{table}

\begin{figure}[t]
\centering
\begin{tikzpicture}
\begin{axis}[
  ybar, width=\columnwidth, height=4cm, bar width=2.5pt,
  ylabel={\small IPC change (\%)},
  ylabel style={at={(0.05,0.5)}},
  symbolic x coords={gcc,mcf,pop2,omn,bwa,xal,cam,per,x264,cac,
                     lbm,wrf,img,fot,rom,exc,dsj,nab,xz,lee},
  xtick=data,
  x tick label style={rotate=60, anchor=east, font=\tiny},
  y tick label style={font=\tiny},
  ymin=-5, ymax=0.5,
  ytick={-4,-3,-2,-1,0},
  extra y ticks={0}, extra y tick style={grid=major, grid style={black,thick}},
  grid=major, y grid style={dashed, gray!40},
  every axis plot/.append style={fill opacity=0.85},
  clip=false,
]
\addplot[fill=black!50, draw=black!70] coordinates {
  (gcc,-4.37) (mcf,-2.92) (pop2,-0.37) (omn,-0.34) (bwa,-0.26) (xal,-0.24)
  (cam,-0.03) (per,-0.02) (x264,-0.02) (cac,-0.01)
  (lbm,0.00) (wrf,0.00) (img,0.00) (fot,0.00) (rom,0.00)
  (exc,0.01) (dsj,0.03) (nab,0.04) (xz,0.06) (lee,0.08)
};
\end{axis}
\end{tikzpicture}
\caption{Per-program IPC change, gated stride vs.\ raw stride (sorted).
Only GCC ($-4.4\%$) and MCF ($-2.9\%$) are outliers; the remaining 18
programs fall within ${\pm}0.4\%$.}
\label{fig:champsim_ipc}
\end{figure}
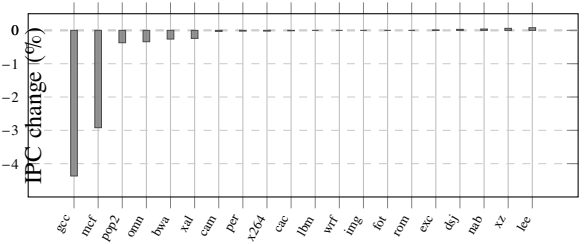

\Cref{tab:champsim} gives the full twenty-program result. The gate removes
35.1\% of accepted L1D prefetches (60.5\% median per-trace) and improves
accuracy by 4.4 percentage points while retaining 90.3\% of useful
prefetches. The apparent improvement shrinks as it descends the hierarchy:
8.4\% and 7.6\% fewer L2C and LLC prefetch misses, and only 0.07\% fewer
DRAM reads. Raw stride achieves $1.044\times$ no-prefetch IPC ($p = 0.096$);
gated stride achieves $1.040\times$ ($p = 0.004$). The direct gate/raw
geometric-mean IPC is $0.996$ (bootstrap 95\% CI: $[0.990,\,0.9997]$; exact
sign test $p = 0.481$).

\subsection{Per-Program Analysis}

\Cref{fig:champsim_ipc} shows the per-program distribution. Two outliers
dominate: GCC loses 4.37\% and MCF loses 2.92\%---both high-volume prefetch
benchmarks where even ``useless'' prefetches may warm the hierarchy. The
remaining eighteen programs have ratios between 0.996 and 1.001. Across twenty
programs there are seven wins, eleven losses, and two ties; the exact
sign-test $p = 0.481$ (not significant). The bootstrap interval
$[0.990,\,0.9997]$ reflects the outliers' weight; the sign test reflects
inconsistent direction. The aggregate DRAM read change ($-0.07\%$) is well within the frozen
${\pm}1\%$ descriptive margin; seventeen of twenty programs change by
less than $0.5\%$, with the largest per-program change $-2.1\%$
(perlbench, on a small absolute base of ${\sim}10.6$K DRAM reads).

Notably, the gated prefetcher achieves a \emph{stronger} result against the
no-prefetch baseline than raw stride does: $1.040\times$ with $p = 0.004$
(14~wins, 2~losses, 4~ties) versus $1.044\times$ with $p = 0.096$
(13~wins, 5~losses, 2~ties). The gate trades a marginal amount of peak
speedup for substantially more consistent directionality.

\subsection{\texorpdfstring{Interpretation: Proxy $\neq$ Endpoint}{Interpretation: Proxy ≠ Endpoint}}

The correct interpretation is not ``gating saves memory bandwidth.'' Most
suppressed prefetches either do not descend the hierarchy or return later as
demand reads. Accuracy and L1D issue count are useful diagnostics, but they
are not endpoint evidence. Better proxy metrics (fewer prefetches, higher
accuracy) establish neither lower off-chip traffic nor higher performance.

The twenty-program data sharpens this finding. A 35\% aggregate (60\%
median) reduction in issued prefetches yields only 0.07\% fewer DRAM reads.
L2C and LLC prefetch misses drop by 8\%, but this does not translate to
measurable IPC gain. The proxy-to-endpoint gap is the second major finding:
proxy metrics must not substitute for endpoint validation.
This is independently corroborated: the Hermes neural off-chip
predictor~\cite{hermes} increases DRAM transactions by 5--10\% despite
reducing latency, confirming that proxy gains can coexist with endpoint
regression. Mahling et al.~\cite{mahling2025} further document how
hardware prefetching impact varies sharply across in-memory workload
types, reinforcing the need for endpoint-specific validation.

\section{Discussion}
\label{sec:discussion}

\subsection{The Gate Matters More Than the Predictor}

The matched-control experiment (\cref{sec:falsification}) and the ChampSim
validation (\cref{sec:champsim}) converge on one conclusion: the gate is
the mechanism that matters. When evaluated without a gate, the MLP appears
safer than stride on random traffic---but with matched admission, the
difference vanishes. The gate itself is useful: it suppresses a third of
prefetches and improves accuracy. But the MLP behind it is not competitive
with stride on most learnable streams, and the gate's proxy-metric
improvements do not reliably translate to endpoint gains.

\subsection{The Dense-Compute Limitation}

The \texttt{real\_matmul} and \texttt{real\_mergesort} results
(\cref{tab:mlagent}) deserve analysis. On these workloads the access
pattern is a large, fixed stride (32~words for row-major matrix multiply).
Classical prefetchers recognize this pattern instantly; the best
achieves up to 11$\times$ speedup. The MLP's replay buffer, clamped to ${\pm}16$~words,
cannot represent a 32-word stride. The gate correctly identifies unreliable
predictions and suppresses prefetching---safety working as designed, at the
cost of missing workloads where a simple heuristic suffices.

The architectural fix is a stride-aware bypass: detect stable large strides
and delegate to a stride unit without consulting the MLP. This
hybrid---learned predictor for complex patterns, heuristic fallback for
simple ones---is the planned next step.

\subsection{When Admission Control Helps vs.\ Hurts}

Admission control is not universally beneficial. In the local model, the
regularity gate improves performance on random ($1.089\times$), phased
($1.076\times$), and agent-RAG ($1.016\times$) streams relative to raw
stride, with reductions of 99.55\%, 83.07\%, and 20.86\% in accepted
prefetches. But at ideal lookahead $D = 1$ (a local-simulator setting distinct
from the $D = 16$ used in \cref{tab:matched}) the gate slows agent RAG by
1.30\%, and on one deterministic merge-sort trace it is 5.95\% slower.

The twenty-program ChampSim study refines the picture. The gate improves
\emph{directional consistency}: against the no-prefetch baseline, gated
stride wins on 14 of 20 programs ($p = 0.004$) while raw stride wins on
only 13 ($p = 0.096$). The gate converts three raw-stride losses into wins
or ties, at the cost of slight regression on two high-volume benchmarks
(GCC, MCF). The aggregate IPC cost is 0.42\%, with a bootstrap CI upper bound
within 0.05\% of parity ($[0.990,\,0.9997]$).

The trade-off is between prefetch volume and accuracy on one hand, and
performance on the other. Reducing prefetches is not intrinsically valuable
unless the freed bandwidth or cache capacity translates to lower latency.
The gate's value lies in consistency and reduced pollution, not in raw speed.

\subsection{Implications for LLM Inference Prefetching}

The motivating scenario---mixed-phase LLM inference workloads---remains
the most promising target for confidence-gated admission. KV-cache appends
are sequential and predictable; attention gathers are irregular. A gate
that opens on the former and closes on the latter avoids cache pollution
from the irregular phase without sacrificing the sequential benefit. The
local simulator confirms this: the gated NN achieves $1.33\times$ on
KV-cache append while maintaining $1.000\times$ on random.

However, the ChampSim results caution against overconfidence: SPEC CPU2017
programs do not exhibit this mixed-phase structure, and the gate's benefit
did not materialize as endpoint improvement on those workloads. Validating
the gate on real LLM inference traces---with native or RTL-level
simulation---is the critical next step for this application.

\section{Threats to Validity}
\label{sec:threats}

\textbf{Internal validity.} The local simulator is deterministic given a
seed, eliminating measurement noise. The main threat is the simulator's
abstraction level: results reflect modeled latencies, not real silicon. We
mitigate this by testing across a range of machine parameters
(\cref{fig:sensitivity}) and by validating externally with
ChampSim (\cref{sec:champsim}).

\textbf{External validity.} The native study uses one highest-weight
SimPoint per program for twenty SPEC CPU2017 programs (eleven development,
nine holdout), not all SimPoints; one ChampSim configuration; and no
multicore interference. SPEC CPU2017 does not validate the workload-shaped
ML/agent streams. Local workloads are synthetic or captured from simple
algorithms; production traces may expose patterns where the gate behaves
differently.
The native ChampSim study validates the gate mechanism itself (gated vs.\
raw stride) but does not replicate the MLP-vs-stride matched control of
\cref{sec:falsification} externally; extending \cref{tab:matched}'s
comparison to native SPEC CPU2017 traces is left to future work, and
readers should not extrapolate the local matched-control finding to the
native setting without that validation.

\textbf{Construct validity.} We use speedup and L1 hit rate as proxies for
prefetcher quality locally, and IPC plus DRAM reads as endpoints in
ChampSim. The ChampSim results demonstrate precisely the risk of proxy
metrics: improved accuracy did not produce improved IPC. Energy, bandwidth
waste, and multi-tenant interference are not captured.
Additionally, our matched-control result concerns a small (257-parameter)
online delta-regressor trained from scratch on each trace; whether the
finding generalizes to larger, offline-pretrained architectures such as
Pythia~\cite{pythia}, TransFetch~\cite{transfetch}, or DART~\cite{dart}
remains an open question that this paper does not address.
Finally, our matched-control comparisons gate only the stride predictor;
next-line and best-offset in \cref{tab:mlagent} remain ungated. The claim
that the gate is useful independent of the underlying predictor is
therefore supported by one classical baseline (stride), not by a broader
sweep of gated classical predictors.

\textbf{Statistical conclusion validity.} All local significance tests are
paired Wilcoxon signed-rank (non-parametric). We do not apply Bonferroni
correction for multiple comparisons; with 5 workloads $\times$ 4
comparisons the adjusted $\alpha$ would be $0.05/20 = 0.0025$, and our
reported $p$-values supporting our primary claims ($< 0.001$) survive
this correction; the practically negligible residual on \texttt{random\_read}
($p = 0.008$) does not, consistent with its negligible effect size ($< 0.03\%$ of cycles). ChampSim results
use exact sign tests with Holm correction over the comparison family.

\textbf{Implementation validity.} Custom cache-callback counts in ChampSim
can differ by up to 0.105\% even with identical retired-instruction counts,
because ChampSim changes phase at core retirement while requests remain in
flight. We report this diagnostic but do not use it as an experimental unit.
Exact raw/official endpoint equivalence remains the hard invariant.

\textbf{Model capacity.} The MLP has only 257~parameters; a larger network
might close the gap with stride on learnable traffic. However, our
finding---that the gate, not the predictor, produces the safety
property---is capacity-independent; whether the relative performance
of MLP vs.\ gated stride would change with a larger model remains open.

\begin{figure}[t]
\centering
\includegraphics[width=\columnwidth]{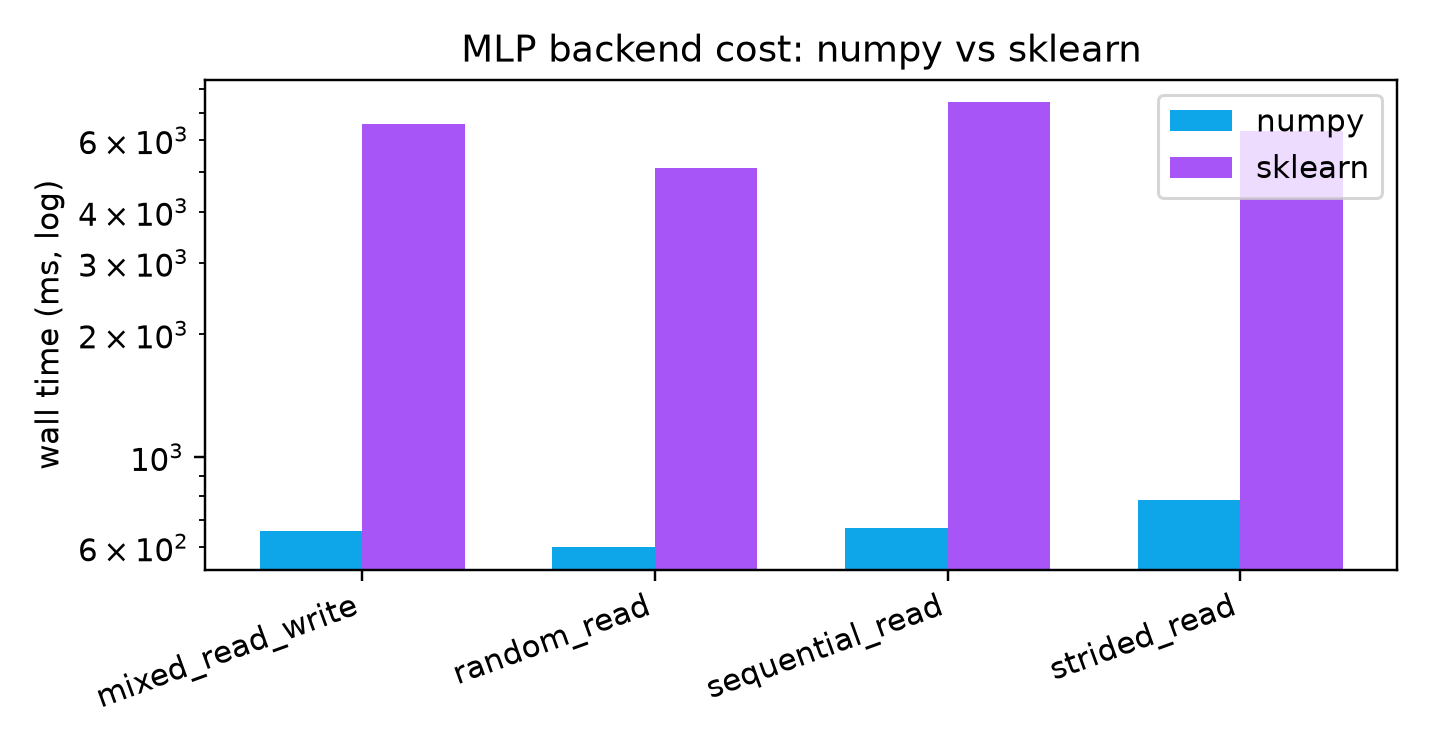}
\caption{Backend wall-time comparison: pure-NumPy vs.\ scikit-learn MLP at
statistically identical hit rates.}
\label{fig:backend}
\end{figure}

\begin{figure}[t]
\centering
\includegraphics[width=\columnwidth]{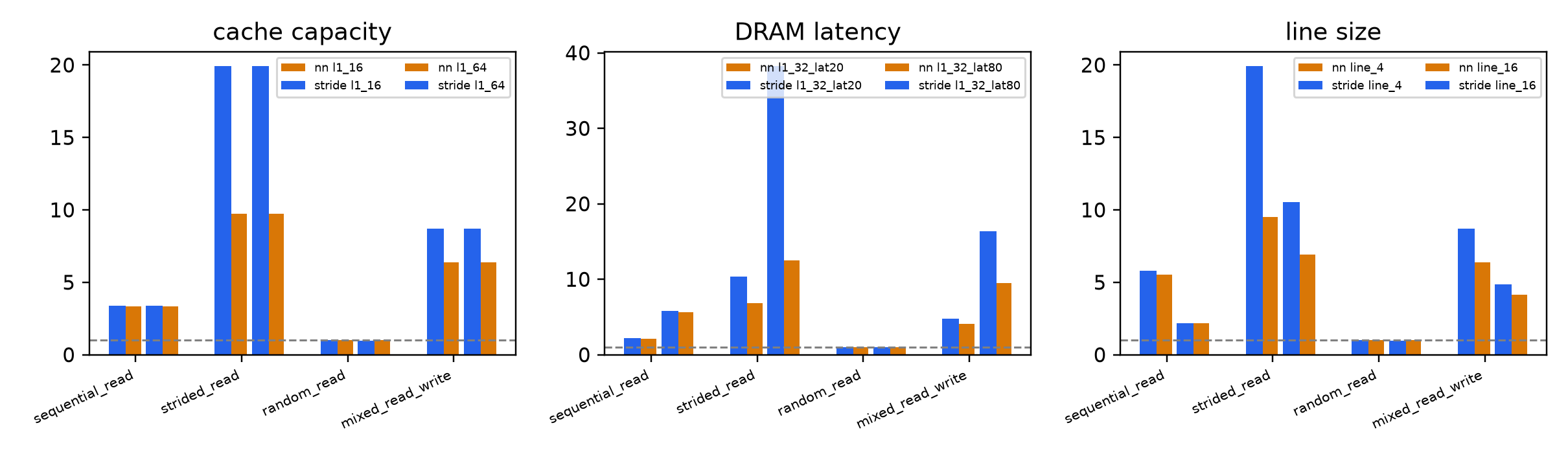}
\caption{Machine sensitivity: NN vs.\ stride speedup for two settings of
cache capacity, DRAM latency, and line size. The NN's safety on random
persists across all settings.}
\label{fig:sensitivity}
\end{figure}

\section{Conclusion}
\label{sec:conclusion}

Three findings emerge. First, comparing a gated learned prefetcher against
an ungated classical baseline confounds prediction with admission. The
neural advantage on random traffic is an admission-policy effect; future
evaluations of learned prefetchers should use matched controls before
attributing performance to the model.

Second, the confidence gate is architecturally useful independent of the
predictor---but the twenty-program ChampSim study reveals a measurement
lesson: a 35\% prefetch reduction (60\% median) yields only $0.07\%$ fewer
DRAM reads. Proxy improvements must not substitute for endpoint validation.

Third, the gate's practical value is directional consistency: $14/20$
programs beat no-prefetch ($p = 0.004$) versus $13/20$ for raw stride
($p = 0.096$), trading $0.4\%$ peak IPC for fewer programs where
prefetching hurts.
We note this directional pattern is descriptive rather than statistically
conclusive: the direct gate-vs-raw-stride comparison itself is not
significant (geometric-mean IPC ratio $0.996$, exact sign test
$p = 0.481$). We report the win/loss split as suggestive of improved
consistency, not as an independently established effect.

The path forward is a better gate---handling dense large-stride patterns
and adapting to workload phase---validated on LLM inference traces with
RTL-level energy characterization.

\section*{Reproducibility Statement}
\label{sec:reproducibility}

Every experiment persists \texttt{config.json} and \texttt{manifest.json}
(git revision, library versions, host) next to raw per-seed
\texttt{runs.csv} and aggregated \texttt{summary.csv}. One command reruns
the full local battery (\texttt{python benchmarks/paper\_experiments.py}).
Machine-verifiable checks:
{\small\begin{verbatim}
python -m benchmarks.verify_paper
python -m benchmarks.verify_research
python champsim/verify_results.py \
  --results results/champsim_full_20
\end{verbatim}}
The last command reconstructs every native aggregate from the
\texttt{champsim\_full\_20} artifact (20~traces $\times$ 4~configs = 80~raw runs),
checks official-control equivalence, and evaluates the frozen predictions.
The local artifact (\texttt{results/research\_gate\_validation/runs.csv})
has 21,240~raw rows. The local and native artifacts each record source,
configuration, unit, and result provenance. Native metadata additionally
pins the ChampSim revision and binary, module, trace-list, and trace hashes.
The local MLP backend (pure NumPy) achieves statistically identical hit
rates to a scikit-learn reference implementation (\cref{fig:backend}).
All code, data, and verification scripts are publicly available at
\url{https://github.com/algosolltd/nncpu}.


\begin{thebibliography}{30}

\bibitem{wulf1995}
W.~A. Wulf and S.~A. McKee,
``Hitting the memory wall: Implications of the obvious,''
\emph{ACM SIGARCH Computer Architecture News}, vol.~23, no.~1, pp.~20--24, 1995.

\bibitem{chen1995}
T.-F. Chen and J.-L. Baer,
``Effective hardware-based data prefetching for high-performance processors,''
\emph{IEEE Trans.\ Computers}, vol.~44, no.~5, pp.~609--623, 1995.

\bibitem{kim2016spp}
J.~Kim, S.~H. Pugsley, P.~V. Gratz, A.~L.~N. Reddy, C.~Wilkerson, and Z.~Chishti,
``Path confidence based lookahead prefetching,''
in \emph{Proc.\ MICRO}, 2016, pp.~1--12.

\bibitem{michaud2016bop}
P.~Michaud,
``Best-offset hardware prefetching,''
in \emph{Proc.\ HPCA}, 2016, pp.~469--480.

\bibitem{bingo}
M.~Bakhshalipour, M.~Shakerinava, P.~Lotfi-Kamran, and H.~Sarbazi-Azad,
``Bingo spatial data prefetcher,''
in \emph{Proc.\ HPCA}, 2019, pp.~399--411.

\bibitem{hashemi2018}
M.~Hashemi, K.~Swersky, J.~Smith, G.~Ayers, H.~Litz, J.~Chang,
C.~Kozyrakis, and P.~Ranganathan,
``Learning memory access patterns,''
in \emph{Proc.\ ICML}, 2018.

\bibitem{pythia}
R.~Bera, K.~Kanellopoulos, A.~Nori, T.~Shahroodi, S.~Subramoney, and O.~Mutlu,
``Pythia: A customizable hardware prefetching framework using online reinforcement learning,''
in \emph{Proc.\ MICRO}, 2021.
\newblock DOI: 10.1145/3466752.3480114.

\bibitem{transfetch}
P.~Zhang, A.~Srivastava, A.~Nori, R.~Kannan, and V.~K. Prasanna,
``Fine-grained address segmentation for attention-based variable-degree prefetching,''
in \emph{Proc.\ Computing Frontiers (CF)}, 2022.
\newblock DOI: 10.1145/3528416.3530236.

\bibitem{dart}
P.~Zhang, N.~Gupta, R.~Kannan, and V.~K. Prasanna,
``Attention, distillation, and tabularization: Towards practical neural
network-based prefetching,''
in \emph{Proc.\ IPDPS}, 2024.
\newblock DOI: 10.1109/IPDPS57955.2024.00082.

\bibitem{pathfinder}
L.~Jia, J.~P. McMahon, S.~Gudaparthi, S.~Singh, and R.~Balasubramonian,
``{Pathfinder}: Practical real-time learning for data prefetching,''
in \emph{Proc.\ ASPLOS}, 2024.
\newblock DOI: 10.1145/3620666.3651332.

\bibitem{champsim}
N.~Gober \emph{et al.},
``The Championship Simulator: Architectural simulation for education and competition,''
\emph{arXiv preprint arXiv:2210.14324}, 2022.

\bibitem{dpc3}
``Third Data Prefetching Championship public traces,''
\url{https://dpc3.compas.cs.stonybrook.edu/champsim-traces/speccpu/}.

\bibitem{kwon2023vllm}
W.~Kwon \emph{et al.},
``Efficient memory management for large language model serving with {PagedAttention},''
in \emph{Proc.\ SOSP}, 2023.

\bibitem{pope2023}
R.~Pope \emph{et al.},
``Efficiently scaling transformer inference,''
in \emph{Proc.\ MLSys}, 2023.

\bibitem{kvcache_streamed}
Y.~Sheng \emph{et al.},
``{FlexGen}: High-throughput generative inference of large language models with a single {GPU},''
in \emph{Proc.\ ICML}, 2023.

\bibitem{kvcache_quant}
J.~Lin \emph{et al.},
``{AWQ}: Activation-aware weight quantization for {LLM} compression and acceleration,''
in \emph{Proc.\ MLSys}, 2024.

\bibitem{orca}
G.~Yu \emph{et al.},
``{Orca}: A distributed serving system for transformer-based generative models,''
in \emph{Proc.\ OSDI}, 2022.

\bibitem{sarathi}
A.~Agrawal \emph{et al.},
``Taming throughput-latency tradeoff in {LLM} inference with {Sarathi-Serve},''
in \emph{Proc.\ OSDI}, 2024.

\bibitem{leviathan2023}
Y.~Leviathan, M.~Kalman, and Y.~Matias,
``Fast inference from transformers via speculative decoding,''
in \emph{Proc.\ ICML}, 2023.

\bibitem{kvcache_lookahead}
Y.~Zhao, Z.~Xie, C.~Liang, C.~Zhuang, and J.~Gu,
``Lookahead: An inference acceleration framework for large language models
with lossless generation accuracy,''
\emph{arXiv preprint arXiv:2312.12728}, 2023.

\bibitem{apex}
A.~Kanani, L.~Badawi, and U.~Y. Ogras,
``{APEX}: Adaptive expert prefetching for memory-efficient edge {MoE} inference,''
\emph{arXiv preprint arXiv:2608.11688}, 2026.

\bibitem{kvcache_moe}
Y.~Zhao, R.~Bunescu, A.~Louri, A.~Karanth, and K.~Wang,
``A spatio-temporal expert prefetching framework for efficient {MoE}-based
{LLM} inference,''
\emph{arXiv preprint arXiv:2606.15453}, 2026.

\bibitem{kraska2018}
T.~Kraska, A.~Beutel, E.~H.~Chi, J.~Dean, and N.~Polyzotis,
``The case for learned index structures,''
in \emph{Proc.\ ACM SIGMOD}, pp.~489--504, 2018.

\bibitem{laroche2019}
R.~Laroche, P.~Trichelair, and R.~Tachet~des~Combes,
``Safe policy improvement with baseline bootstrapping,''
in \emph{Proc.\ ICML}, pp.~3652--3661, 2019.

\bibitem{berti}
A.~Navarro-Torres, B.~Panda, J.~Alastruey-Bened\'{e}, P.~Ib\'{a}\~{n}ez,
V.~Vi\~{n}als-Y\'{u}fera, and A.~Ros,
``{Berti}: An accurate local-delta data prefetcher,''
in \emph{Proc.\ MICRO}, 2022, pp.~975--991.
\newblock DOI: 10.1109/MICRO56248.2022.00072.

\bibitem{sglang}
L.~Zheng \emph{et al.},
``{SGLang}: Efficient execution of structured language model programs,''
in \emph{Proc.\ NeurIPS}, 2024.

\bibitem{ppf}
E.~Bhatia, G.~Chacon, S.~Pugsley, E.~Teran, P.~V. Gratz, and D.~A. Jim\'{e}nez,
``Perceptron-based prefetch filtering,''
in \emph{Proc.\ ISCA}, 2019, pp.~1--13.
\newblock DOI: 10.1145/3307650.3322207.

\bibitem{hermes}
R.~Bera \emph{et al.},
``{Hermes}: Accelerating long-latency load requests via perceptron-based
off-chip load prediction,''
in \emph{Proc.\ MICRO}, 2022.
\newblock DOI: 10.1109/MICRO56248.2022.00015.

\bibitem{mahling2025}
F.~Mahling, M.~Weisgut, and T.~Rabl,
``Fetch me if you can: Evaluating {CPU} cache prefetching and its reliability
on high latency memory,''
in \emph{Proc.\ DaMoN}, 2025.
\newblock DOI: 10.1145/3736227.3736231.

\bibitem{fdp2007}
S.~Srinath, O.~Mutlu, H.~Kim, and Y.~N. Patt,
``Feedback directed prefetching: Improving the performance and
bandwidth-efficiency of hardware prefetchers,''
in \emph{Proc.\ HPCA}, 2007, pp.~63--74.
\newblock DOI: 10.1109/HPCA.2007.346185.

\end{thebibliography}
\end{document}